\documentclass[preprint,12pt]{elsarticle}

\usepackage{graphicx,mathrsfs}
\usepackage{amsmath,amsfonts,amsthm,amssymb}
\usepackage{hyperref,xcolor,subfig}

\newtheorem{tlem}{Lemma}
\newtheorem{tthm}[tlem]{Theorem}

\newtheorem{tcor}[tlem]{Corollary}

\def \ie {i.e.~}

\journal{Computers \& Operations Research}

\begin{document}

\begin{frontmatter}

\title{A DSATUR-based algorithm for the Equitable Coloring Problem\tnoteref{grant}}

\author[a]{Isabel M\'endez-D\'iaz\fnref{correspon}}
\author[b,c]{Graciela Nasini}
\author[b,c]{Daniel Sever\'in}

\address[a]{ FCEyN, Universidad de Buenos Aires,
	Argentina }

\address[b]{ FCEIA, Universidad Nacional de Rosario,
	Argentina }

\address[c]{ CONICET, Argentina }

\tnotetext[grant]{Partially supported by grants PIP-CONICET 241, PICT 2011-0817 and UBACYT 20020100100666. 
\emph{E-mail addresses}: \texttt{imendez@dc.uba.ar} (I. M\'endez-D\'iaz),
\texttt{nasini@fceia.unr.edu.ar} (G. Nasini), \texttt{daniel@fceia.unr.edu.ar} (D. Sever\'in).}
\fntext[correspon]{Corresponding author at Departamento de Ciencias de la Computaci\'on, Facultad de Ciencias Exactas y Naturales, Universidad de Buenos Aires,
Intendente Guiraldes 2160 (Ciudad Universitaria, Pabell\'on 1), Argentina}

\begin{abstract}
This paper describes a new exact algorithm for the Equitable Coloring Problem, a coloring problem
where the sizes of two arbitrary color classes differ in at most one unit.
Based on the well known \textsc{DSatur} algorithm for the classic Coloring Problem,
a pruning criterion arising from equity constraints is proposed and analyzed.
The good performance of the algorithm is shown through computational experiments over random and benchmark instances.
\end{abstract}


\end{frontmatter}


\section{Introduction} \label{SINTRO}

There exists a large family of combinatorial optimization
problems having relevant practical importance, besides its theoretical
interest. One of the most representative problem of this family is the
\emph{Graph Coloring Problem} (GCP), which arises in many applications
such as scheduling, timetabling, electronic bandwidth allocation and sequencing problems.

Given a simple graph $G = (V, E)$, where $V$ is the set of vertices and $E$ is the set of edges, a \emph{coloring of $G$} is an assignment of colors to
vertices such that the endpoints of any edge have different colors.
A \emph{$k$-coloring of $G$} is a coloring that uses $k$ colors.
The GCP consists of finding the minimum number $k$ such that $G$ admits a $k$-coloring.
This minimum number of colors is called the \emph{chromatic number} of $G$ and is denoted by $\chi(G)$.\\

It is well known that GCP models some scheduling problems. The simplest version considers assignments of workers to a given set of tasks.
Pairs of tasks may conflict each other, meaning that they should not be assigned to the same worker.
The problem is modeled by building a graph containing a vertex for every task and an edge for every conflicting pair of tasks. A coloring of this graph represents a conflict-free assignment and the chromatic
number of the graph is exactly the minimum number of workers needed to perform all tasks.

However, an extra constraint could be required to ensure the uniformity of the distribution of workload employees.
The addition of this extra \emph{equity} constraint gives rise to the
\emph{Equitable Coloring Problem} (ECP), introduced in \cite{MEYER} and motivated by an application concerning \emph{garbage collection} \cite{EXAMPLE2}. 
Other applications of the ECP concern \emph{load balancing problems} in multiprocessor machines \cite{EXAMPLE3}
and results in \emph{probability theory} \cite{EXAMPLE1}. An introduction to ECP and some basic results are provided in
\cite{KUBALE}.

Formally, an \emph{equitable $k$-coloring} (or just $k$-eqcol) of a graph $G$ is a $k$-coloring 
satisfying the \emph{equity constraint}, \ie the size of two color classes can not differ by more than one unit.
The \emph{equitable chromatic number} of $G$, $\chi_{eq}(G)$, is the
minimum $k$ for which $G$ admits a $k$-eqcol. The ECP consists of finding $\chi_{eq}(G)$.

Computing $\chi_{eq}(G)$ for arbitrary graphs is proved to be NP-Hard and
just a few families of graphs are known to be easy such as complete $n$-partite, complete split, wheel and tree graphs \cite{KUBALE}.

There exist some differences between GCP and ECP that make the latter harder to solve.
It is known that the chromatic number of an unconnected graph $G$ is the maximum among the chromatic numbers of its components.
Algorithms that solve GCP can take advantages of the property mentioned above (e.g. \cite{BCCOLBRANCHINGRULE}) by solving GCP
on each component, which is less CPU intensive than address the problem on the whole graph.
Moreover, one can preprocess the graph in order to reduce its size and, consequently, the time of optimization. For example, choosing two non-adjacent vertices with the same neighborhood, known as twin vertices, and deleting one of them. The chromatic number of the graph
remains the same after deletion, since the deleted vertex can inherit the color of the other one.
None of these recipes can be applied when solving ECP. For instance, let $G$ be the graph of Figure $1a$ and $G'$ be the graph compounded of two disjoint copies
of $G$. Then, $\chi_{eq}(G') = 2$ but $\chi_{eq}(G) = 3$. Also, let $H'$ be the graph of Figure $1b$. Clearly, $v$ and $v'$ are twin vertices. Let $H'$ be
$H$ after $v$ is deleted. We have $\chi_{eq}(H') = 2$ but $\chi_{eq}(H) = 3$.\\

\begin{figure}[h]
\begin{center}
\includegraphics[scale=0.4]{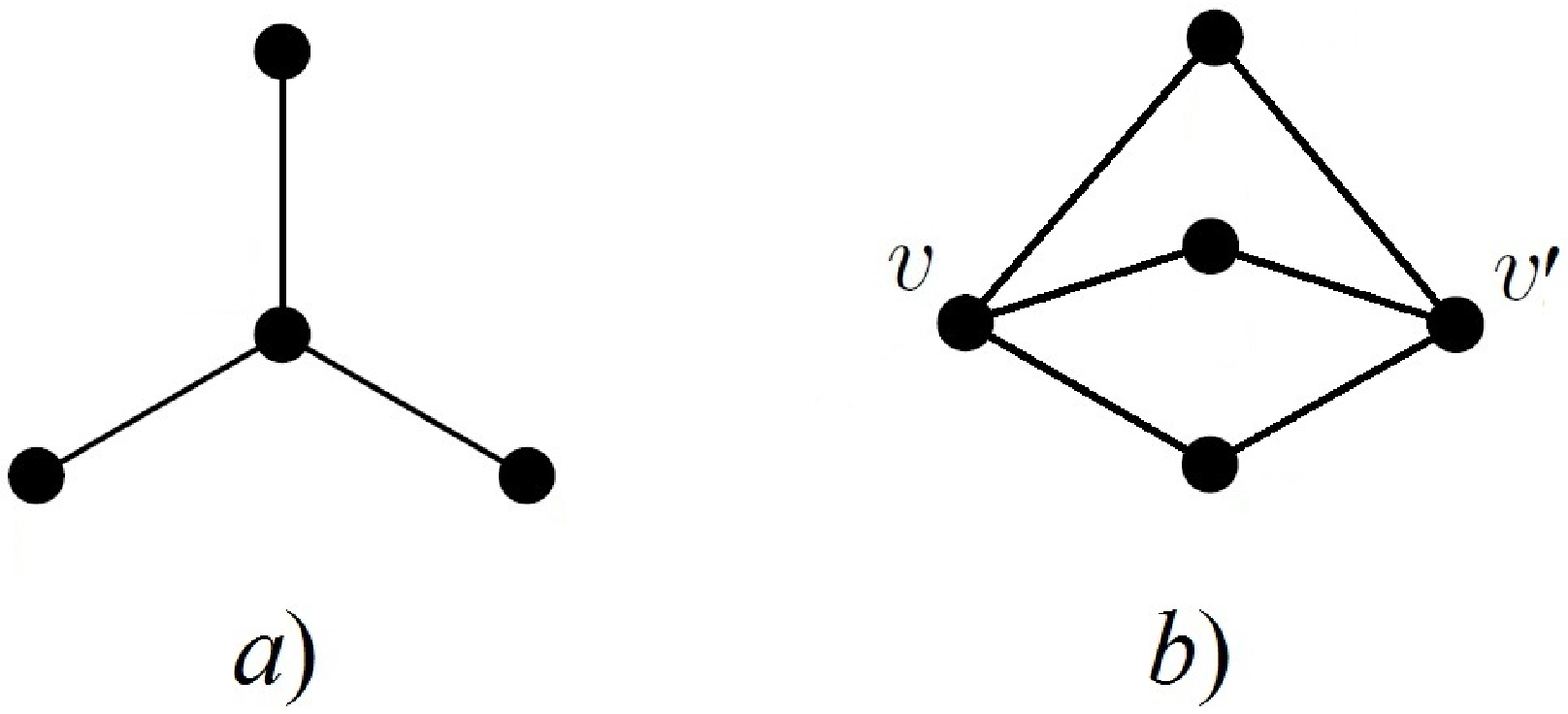}
\end{center}
\vspace{-20pt}
\caption{}
\label{fig:grasu}
\end{figure}

There are very few tools in the literature related to ECP resolution.
Two constructive algorithms called \textsc{Naive} and \textsc{SubGraph} were given in \cite{KUBALE} to generate greedily an equitable coloring of a graph and,
as far as we know, two integer linear programming approaches are available. The first one is a Branch-and-Cut algorithm, called B\&C-$LF_2$
\cite{BYCBRA}, which is based on the asymmetric representatives formulation for GCP described in \cite{REPRESENTATIVES}.
The other one \cite{PAPERDAM} adapts to ECP the formulation and techniques used by M\'endez-D\'iaz and Zabala for GCP in \cite{BCCOLBRANCHINGRULE},
studies its polyhedral structure and derives families of valid inequalities. Some of them have shown to be very effective as cutting planes in preliminary
computational experiments.

Regarding GCP, we can find good exact algorithms which are not based on IP techniques. One of the most well known example is \textsc{DSatur}, proposed by Br\'elaz in \cite{DSATUR}. This Branch-and-Bound algorithm has been referred in the literature several times and is still used by its simplicity, its efficiency in
medium-sized graphs and the possibility of applying it at some stage in metaheuristics or in more complex
exact algorithms like Branch-and-Cut ones \cite{BCCOLBRANCHINGRULE}. Recently, it was shown that a modification of \textsc{DSatur} performs relatively well
compared with many state-of-the-art algorithms based on IP techniques, showing superiority in random instances \cite{PASS}.

This fact encourages us to research how to modify a DSatur-based solver in order to address the ECP, which is the goal of this paper.
Our approach exploits arithmetical properties inherent in equitable colorings and combines them with
the techniques originally developed by
Brown \cite{BROWN} and Br\'elaz \cite{DSATUR} for \textsc{DSatur}, and improved by Sewell \cite{SEWELL} and San Segundo \cite{PASS}.
We call it \textsc{EqDSatur}.
A preliminary version of this algorithm with weaker pruning rules than the one analyzed in this work was already presented in \cite{LAGOS2013}.

The paper is organized as follows.
Section \ref{SDSATUR} gives a brief summary of known DSatur-based algorithms for GCP.
Section \ref{SNEWPRUN} shows the background math for our pruning rule.
Section \ref{SEQIMPL} describes an implementation of \textsc{EqDSatur}.
Section \ref{SBOUNDS} discusses methods for obtaining lower and upper bounds of the equitable chromatic number.
Section \ref{SCOMPU} reports computational experiments carried out to tune up the behaviour of \textsc{EqDSatur},
and compares our algorithm against other ones from the literature. Finally, Section \ref{SCONCLU} gives final conclusions.\\

We now introduce some notations and definitions employed throughout the paper. For any positive integer $k$, $[k]$ denotes the set $\{1,2,\ldots,k\}$. Given a graph $G=(V,E)$, we assume the set of vertices is $V = [n]$. A graph for which every vertex is adjacent to each other is called a \emph{complete graph}.
Given $S \subset V$, we denote by $G[S]$ the subgraph of $G$ \emph{induced by} $S$. A set $Q \subset V$ is a \emph{clique} of $G$ if $G[Q]$ is a complete graph.

Given $u \in V$, the \emph{neighborhood} of $u$ is the set of vertices adjacent to $u$ and is denoted by $N(u)$. The
\emph{closed neighborhood} of $u$, $N[u]$, is the set $N(u) \cup \{u\}$. The \emph{degree of $u$}, $d(u)$, is the cardinality of $N(u)$.
The maximum degree of vertices in $G$ is denoted by $\Delta(G)$.

A \emph{stable set} is a set of vertices of $G$ no two of which are adjacent.
We denote by $\alpha(G)$ the \emph{stability number} of $G$, \ie the maximum cardinality of a stable set of $G$.
Given $S \subset V$, we also denote by $\alpha(S)$ the stability number of $G[S]$.

A \emph{partial $k$-partition} of $G$, denoted by $\Pi=(C_1,C_2,\ldots,C_n)$, is a collection of disjoint sets such that
$\cup_{j=1}^k C_j \subset V$ and $C_j = \varnothing$ if and only if $j \geq k+1$. We write $k(\Pi)$ to refer the number of non-empty sets
in $\Pi$. We denote by $U(\Pi)$ the set of vertices not covered by the sets of $\Pi$, \ie $U(\Pi) = V \backslash\!\cup_{j=1}^k C_j$.
If $U(\Pi) = \varnothing$ we say that $\Pi$ is a \emph{$k$-partition}. Given $v \in V \backslash U$, we denote by $\Pi(v)$ the number
of the set to which $v$ belongs, \ie $v \in C_{\Pi(v)}$. 

A \emph{partial $k$-coloring} of $G$ is a partial $k$-partition $\Pi=(C_1,C_2,\ldots,C_n)$ of $G$ such that each $C_j$ is
a stable set of $G$. In this context, $U(\Pi)$ is called the \emph{set of uncolored vertices} of a partial $k$-coloring $\Pi$.
If $U(\Pi) = \varnothing$ we say that $\Pi$ is a \emph{$k$-coloring}.

Given $v \in V$ and a partial $k$-coloring $\Pi$, let $D_{\Pi}(v)$ be the set of different colors assigned to the adjacent vertices of $v$, \ie
$D_{\Pi}(v) = \{ \Pi(w) : w \in N(v) \backslash U(\Pi) \}$. The \emph{saturation degree} of $v$ in $\Pi$, $\rho_{\Pi}(v)$, is the cardinality of $D_{\Pi}(v)$
and the \emph{set of available colors} of $v$, $F_{\Pi}(v)$, is the set of unused colors in the neighborhood of $v$, \ie $F_{\Pi}(v) = [n] \backslash D_{\Pi}(v)$.

Given a partial $k$-partition $\Pi$, $u \in U(\Pi)$ and $j \in [k+1]$ we denote by $\Pi + \langle u,j \rangle$ to the partial partition
obtained by adding $u$ to $C_j$.

We say that a partial $k$-partition (or partial $k$-coloring) $\Pi = (C_1$, $C_2$, $\ldots$, $C_n)$ \emph{can be extended} to a $k'$-partition (or $k'$-coloring) if there exists a $k'$-partition (or $k'$-coloring) $\Pi' =  (C'_1, C'_2, \ldots, C'_n)$ which can be obtained from $\Pi$ by
succesive applications of the operator ``+''. A direct consequence is that $k \leq k'$ and $C_j \subset C'_j$ for all $j \in [k]$.

We say that a $k$-partition or $k$-coloring $\Pi=(C_1, C_2, \ldots, C_n)$ of $G$ is \emph{equitable} if it satisfies the equity constraint, \ie 
\[ \left| |C_i| - |C_j| \right| \leq 1,~~~~~ \textrm{for}~i, j \in [k]. \]

An equitable $k$-coloring is also called $k$-eqcol for the sake of simplicity.


\section{An overview of DSatur-based algorithms for GCP} \label{SDSATUR}

The idea behind an enumerative algorithm such as \textsc{DSatur} is to determine early whether it is possible to extend a partial
coloring to a proper coloring so that uncolored vertices are painted with available colors. In this way, the
enumerative procedure avoids to explore partial colorings that will not lead to an optimal coloring, and therefore would be needlessly enumerated.

\textsc{DSatur} is based on a generic enumerative scheme proposed by Brown \cite{BROWN}, outlined as follows:\\

\medskip

\noindent \underline{\textsc{Input}}: $G$ a graph, $\Pi_0$ an initial partial coloring of $G$ and $\Pi^*$ an initial coloring of $G$.\\
\\
\noindent \underline{\textsc{Output}}: $\Pi^*$ an optimal coloring of $G$, $UB$ the chromatic number of $G$.\\
\\
\noindent \underline{\textsc{Algorithm}}: Set $UB \leftarrow k(\Pi^*)$. Then, execute $\textsc{Node}(\Pi_0)$.\\
\\
\noindent \underline{\textsc{Node}$(\Pi)$}:\\
\indent \emph{Step 1}. If $U(\Pi) = \varnothing$, set $UB \leftarrow k(\Pi)$, $\Pi^* \leftarrow \Pi$ and return.\\
\indent \emph{Step 2}. Select a vertex $u \in U(\Pi)$.\\
\indent \emph{Step 3}. For each color $j \in [\min\{k(\Pi)+1, UB-1\}]$ such that $j \in F_{\Pi}(u)$:
\indent \indent \indent Set $\Pi' \leftarrow \Pi + \langle u, j \rangle$.\\
\indent \indent \indent If $F_{\Pi'}(v) \cap [UB-1] \neq \varnothing$ for all $v \in U(\Pi')$, execute \textsc{Node}$(\Pi')$.\\

\medskip

The previous scheme only works when the initial partial coloring $\Pi_0$ can be extended to an optimal coloring. A suitable $\Pi_0$ can be
computed as follows: if $Q = \{v_1, v_2, \ldots, v_q\}$ is a maximal clique of $G$, it is known that a $q$-partial coloring $\Pi_0$ such that
$\Pi_0(v_i) = i$ for all $i \in [q]$ can be extended to a $\chi(G)$-coloring.

Indeed, we must know a maximal clique $Q$ and an initial coloring $\Pi^*$ in advance. Moreover, we must state the rule for choosing
vertex $u$ in Step 2 and the order in which colors from $F(u)$ have to be evaluated. From now on, we call to these criteria
\emph{vertex selection strategy} (VSS) and \emph{color selection strategy} (CSS).\\

Br\'elaz proposed the algorithm \textsc{DSatur} \cite{DSATUR} by obtaining a maximal clique $Q$ and an initial coloring $\Pi^*$
with greedy heuristics (one is SLI given in \cite{MATULA} and the other is contributed by himself).
The vertex selection strategy, which we call DSATUR-VSS, selects the uncolored vertex with the largest saturation degree.
In case of a tie, select the vertex with the largest degree. More specifically, let $\rho$ be the maximum saturation degree
of $\Pi$ and $T$ be the so called \emph{set of candidate vertices}:
\[ T = \{ u \in U(\Pi) : \rho_{\Pi}(u) = \rho \}. \] 
DSATUR-VSS chooses $u \in T$ that maximizes $d(u)$.
In the case that more than one vertex in $T$ has the maximum degree, untie them according to some predetermined order, e.g. its number in $V$.

Sewell \cite{SEWELL} suggested a modified tie breaking rule for choosing
$u$ from the set $T$, called \textsc{Celim} (CELIM-VSS). It consists of selecting from the set of vertices tied at maximum saturation degree, the one with the maximum number of common available colors in the neighborhood of uncolored vertices. That is, choose $u \in T$ such that the value 
\[ celim(u) = \sum_{j \in F_{\Pi}(u)} | \{ v \in N(u) \cap U(\Pi) : j \in F_{\Pi}(v) \} | \]
is the highest.

Let us note that, while DSATUR-VSS attempts to estimate future color availability through the degree of vertices, CELIM-VSS also contemplates the impact of coloring a vertex over the uncolored vertices yet. Although CELIM-VSS is more CPU intensive than DSATUR-VSS, fewer nodes are evaluated and,
in the case of medium and high density instances, less time is required to reach the optimality.\\

A further improvement in the vertex selection strategy was recently proposed by San Segundo \cite{PASS}.
The criterion chooses the vertex $u \in T$ that maximizes the value
\[ pass(u) =  \sum_{j \in F_{\Pi}(u)} | \{ v \in N(u) \cap T : j \in F_{\Pi}(v) \} |. \]
By comparing it with Sewell's criterion we may observe that CELIM-VSS minimizes the number of subproblems by systematically reducing available
color at deeper levels of the search tree. By constrast, San Segundo's criterion restricts this computation to the neighbors in the set of tied vertices,
reducing color domains of vertices which are already known to have the least number of available colors, and so therefore more likely to require
a new color at deeper levels of the search tree.

At an early stage of enumeration, the set $T$ has many vertices and the computation of $pass(u)$ induces an overload in the strategy
that, in some cases, worsens the overall performance. In order to prevent this overload, a threshold called $TH$ is introduced by the author.
If $k(\Pi) - \rho \leq TH$, he chooses from the set $T$, the vertex $u$ whose value of $pass(u)$ is the highest. Otherwise,
he chooses the vertex $u$ whose degree is the highest just like DSATUR-VSS. This strategy is called \textsc{Pass} (PASS-VSS).
Several values of this threshold were tested in \cite{PASS} and $TH=3$ was settled as the best option.

This approach proved to be quite competitive with other exact algorithms for GCP from the literature.\\

Regarding the color selection strategy, as far as we know, all DSatur-based implementations merely consider the set of available colors in ascending
order: first evaluate color 1, then color 2, and so on. We call it DSATUR-CSS.\\

Considering the good performance of DSatur-based algorithms for GCP, it is natural to derive an algorithm for ECP consisting of the
previous Brown's scheme by changing the initial coloring in the initialization by an equitable coloring, and checking whether
$\Pi$ is an equitable coloring in Step 1. In summary, this simple algorithm, which we call \textsc{TrivialEqDSatur},
only applies the equity constraint at the leafs of the search tree in the hope that the resulting coloring is equitable. This may cause
\textsc{TrivialEqDSatur} to explore vast regions of the search tree that will not lead to equitable colorings.

Nevertheless, the exploration of useless nodes could be avoided by checking, at each node, whether a partial coloring can be extended to an equitable coloring.
In the next section, we study necessary and sufficient conditions for a partial coloring to be extended to an equitable coloring and how to implement
it as part of a DSatur-based algorithm.


\section{A pruning rule for the ECP} \label{SNEWPRUN}

We now study arithmetical properties of the sizes of color classes in equitable colorings and how to combine them in order to
propose a pruning rule for our algorithm.

From now on, for a partial $k$-partition $\Pi = (C_1, C_2, \ldots, C_n)$, let $M(\Pi)$ be the largest color class in $\Pi$,
$T(\Pi)$ be the index of color classes in $\Pi$ with size $M(\Pi)$, and $t(\Pi)$ be the cardinality of $T(\Pi)$, \ie
$M(\Pi) = \max \{|C_j| : j \in [k]\}$, $T(\Pi) = \{j \in [k] : |C_j| = M(\Pi) \}$ and $t(\Pi) = |T(\Pi)|$.

The following result fully characterizes when a partial partition can be extended to an equitable partition.

\begin{tthm}
Let $\Pi$ be a partial $k$-partition, $M = M(\Pi)$ and $t = t(\Pi)$.
Then, $\Pi$ can be extended to an equitable partition if and only if
\begin{equation} \label{SUPERMRAW}
  n \geq \bigl( M - 1 \bigr) \cdot k + t
\end{equation}
\end{tthm}
\begin{proof}
Clearly, if $\Pi$ can be extended to an equitable partition $\Pi'$, then the classes from $T(\Pi)$ in $\Pi'$ must have at least $M$ vertices.
Consequently, the classes from $[k] \backslash T(\Pi)$ in $\Pi'$ must have at least $M - 1$ vertices.
Then, $n \geq M\cdot t + (M-1)\cdot(k-t)$ which is equivalent to (\ref{SUPERMRAW}).

On the other hand, if (\ref{SUPERMRAW}) holds then $U(\Pi)$ has enough vertices for the following procedure to get an equitable
$k$-partition: add one by one the remaining uncolored vertices to the smallest non-empty class at each step.
\end{proof}

Formula (\ref{SUPERMRAW}) allows us to obtain another way of characterizing equitable colorings besides the traditional definition:

\begin{tcor}
Let $\Pi$ be a $k$-coloring of $G$, $M = M(\Pi)$ and $t = t(\Pi)$. Then, $\Pi$ is a $k$-eqcol if and only if (\ref{SUPERMRAW}) holds.
\end{tcor}
\begin{proof}
By Theorem 1, if (\ref{SUPERMRAW}) holds then $\Pi$ is extended to the equitable $k$-partition $\Pi$ itself. Since $\Pi$ is already a coloring,
$\Pi$ is a $k$-eqcol. The converse is analogous.
\end{proof}

If we wonder when a partial coloring can be extended to an equitable coloring, it is clearly that condition (\ref{SUPERMRAW}) is necessary.
However, if we know a lower bound of $\chi_{eq}$, the condition can be tightened:

\begin{tcor}
Let $\Pi$ be a partial $k$-coloring, $M = M(\Pi)$, $t = t(\Pi)$ and $LB$ be a lower bound of $\chi_{eq}(G)$.
If $\Pi$ can be extended to an equitable coloring, then
\begin{equation} \label{SUPERM}
  n \geq \bigl( M - 1 \bigr)\cdot\max\{k,LB\} + t
\end{equation}
\end{tcor}
\begin{proof}
In the case that $k \geq LB$, (\ref{SUPERM}) holds by Theorem 1. Hence, we assume $k < LB$.
If $\Pi$ can be extended to an equitable $k'$-coloring $\Pi'$,
we have that $k' \geq \chi_{eq}(G) \geq LB$ and classes from $T(\Pi)$ in $\Pi'$ must have at least $M$ vertices.
Consequently, classes from $[LB] \backslash T(\Pi)$ in $\Pi'$ must have at least $M - 1$ vertices. Therefore,
$n \geq M\cdot t + (M-1)\cdot(LB-t)$ and (\ref{SUPERM}) holds.
\end{proof}

We include the condition given in the previous result as a pruning rule in the Brown's scheme. Below, we sketch our approach called \textsc{EqDSatur}:\\

\medskip

\noindent \underline{\textsc{Input}}: $G$ a graph, $\Pi_0$ an initial partial coloring of $G$, $\Pi^*$ an initial equitable coloring of $G$
and $LB$ a lower bound of $\chi_{eq}(G)$.\\
\\
\noindent \underline{\textsc{Output}}: $\Pi^*$ an optimal equitable coloring of $G$, $UB = \chi_{eq}(G)$.\\
\\
\noindent \underline{\textsc{Algorithm}}: Set $UB \leftarrow k(\Pi^*)$. Then, execute $\textsc{Node}(\Pi_0)$.\\
\\
\noindent \underline{\textsc{Node}$(\Pi)$}:\\
\indent \emph{Step 1}. If $U(\Pi) = \varnothing$, set $UB \leftarrow k(\Pi)$, $\Pi^* \leftarrow \Pi$ and return.\\
\indent \emph{Step 2}. Select a vertex $u \in U(\Pi)$.\\
\indent \emph{Step 3}. For each color $j \in [\min\{k(\Pi)+1, UB-1\}]$ such that $j \in F_{\Pi}(u)$:
\indent \indent \indent Set $\Pi' \leftarrow \Pi + \langle u, j \rangle$.\\
\indent \indent \indent If $n \geq \bigl( M(\Pi') - 1 \bigr) \cdot \max\{k(\Pi'),LB\} + t(\Pi')$ and \\
\indent \indent \indent \indent $F_{\Pi'}(v) \cap [UB-1] \neq \varnothing$ for all $v \in U(\Pi')$, execute \textsc{Node}$(\Pi')$.\\

\medskip

The following theorem shows that \textsc{EqDSatur} works:
\begin{tthm}
If $\Pi_0$ can be extended to a $\chi_{eq}(G)$-eqcol then \textsc{EqDSatur} gives the value of $\chi_{eq}(G)$ into the variable $UB$
and an optimal equitable coloring into $\Pi^*$ after its execution.
\end{tthm}
\begin{proof}
In the case that (\ref{SUPERM}) does not hold, the node corresponding to $\Pi'$ is not called since $\Pi'$ can not be extended to an equitable
coloring according to Corollary 3. Therefore, the algorithm does not prune nodes that could reach an optimal equitable coloring.

Also, each coloring reached at Step 1 is indeed an equitable coloring, due to Corollary 2 and the fact that the current coloring satisfies (\ref{SUPERM}).
\end{proof}


\section{Implementation of EqDSatur} \label{SEQIMPL}

It is clear that the scheme proposed previously is barely helpful if we do not know how to implement it in a efficient way.

Below, we propose a detailed fast implementation of \textsc{EqDSatur}. Indentations are meaningful and mark the scope of the operations involved.
All sets listed in the implementation are represented by global binary-valued arrays.
Global variable $k$ is the number of colors of the current partial partition.\\

\medskip

\noindent \underline{\textsc{Input}}: $G$ a graph, $\Pi^*$ an initial eqcol of $G$ and $LB$ a lower bound of $\chi_{eq}(G)$.\\
\\
\noindent \underline{\textsc{Output}}: $\Pi^*$ an optimal eqcol of $G$, $UB = \chi_{eq}(G)$.\\
\\
\noindent \underline{\textsc{Algorithm}}:\\
\indent Set $UB \leftarrow k(\Pi^*)$.\\
\indent Create a partial coloring $\Pi$ such that $C_i \leftarrow \{v_i\}$ for all $i \in [q]$, where\\
\indent \indent $Q = \{v_1, v_2, \ldots, v_q\}$ is a maximal clique of $G$.\\
\indent Set $U(\Pi) \leftarrow V \backslash Q$ and $k \leftarrow q$.\\
\indent Execute $\textsc{Node}(1, q)$.\\
\\
\noindent \underline{\textsc{Node}$(M,t)$}:\\
\indent \emph{Step 1}. If $U(\Pi) = \varnothing$, set $UB \leftarrow k$, $\Pi^* \leftarrow \Pi$ and return.\\
\indent \emph{Step 2}. Select a vertex $u \in U(\Pi)$.\\
\indent \emph{Step 3}. For each $j \in [\min\{k+1, UB-1\}]$ such that $j \in F_{\Pi}(u)$:\\
\indent \indent Set $size \leftarrow |C_j|$.\\
\indent \indent If $j \leq k$, do:\\
\indent \indent \indent If $size = M$, set $t' \leftarrow 1$ and $M' \leftarrow M + 1$.\\
\indent \indent \indent If $size = M-1$, set $t' \leftarrow t + 1$ and $M' \leftarrow M$.\\
\indent \indent \indent If $size \leq M-2$ set $t' \leftarrow t$ and $M' \leftarrow M$.\\
\indent \indent If $j = k+1$, do:\\
\indent \indent \indent If $M = 1$, set $t' \leftarrow t + 1$ and $M' \leftarrow M$.\\
\indent \indent \indent If $M \geq 2$, set $t' \leftarrow t$ and $M' \leftarrow M$.\\
\indent \indent Set $previous\_k \leftarrow k$.\\
\indent \indent Set $k \leftarrow \max\{j,k\}$.\\
\indent \indent If $n \geq \bigl( M' - 1 \bigr) \cdot \max\{k,LB\} + t'$, do:\\
\indent \indent \indent Set $C_j \leftarrow C_j \cup \{u\}$.\\
\indent \indent \indent Set $U(\Pi) \leftarrow U(\Pi) \backslash \{u\}$.\\
\indent \indent \indent Execute \textsc{Node}$(M', t')$.\\
\indent \indent \indent Set $U(\Pi) \leftarrow U(\Pi) \cup \{u\}$.\\
\indent \indent \indent Set $C_j \leftarrow C_j \backslash \{u\}$.\\
\indent \indent Set $k \leftarrow previous\_k$.\\

\medskip

We do not describe implemetation details of how to update $F_{\Pi}(v)$ for the sake of readability,
but it can be found in \cite{PASS}.
On the other hand, details of how to compute the clique $Q$ and the initial equitable coloring is discussed in Section \ref{SBOUNDS}.

It is not hard to see that variables $M$ and $t$ are indeed the cardinality of the largest class and the number of color classes with size $M$ in the
current partial coloring. The update of these variables as well as $U(\Pi)$, $C_j$ and $k$ is performed in constant time.

Updating $M$ and $t$, and checking (\ref{SUPERM}) is cheap but not free. So, it becames important to analyze if the usage of this pruning rule
pays off in terms of CPU time. This task is performed in Section \ref{SCOMPU} through empirical experimentation.


\section{Lower and upper bounds of $\chi_{eq}(G)$} \label{SBOUNDS}

In order to initialize \textsc{EqDSatur}, it is necessary to compute bounds of the equitable chromatic number. In this section, we discuss how to obtain
such values and we report some computational experiments related to them. We remark that, in particular, the lower bound $LB$ remains constant
during the enumeration, so it is essential that the value of $LB$ be as best as possible.

\subsection{Computation of lower bounds}

Clearly, every equitable coloring of $G$ is also a classic coloring of $G$ so every lower bound of $\chi(G)$ can be used as
a lower bound of $\chi_{eq}(G)$. In particular, the size of any maximal clique of $G$ is a known lower bound of $\chi(G)$ and
$\chi_{eq}(G)$. There are several ways suggested in the literature to obtain such cliques. The easiest method is, for a given
graph $G$ and a given vertex $v$, a greedy algorithm that includes $v$ as the first vertex of the clique and then selects
the vertex adjacent to the clique with highest degree in each step until no more vertices can be added to the clique.
Furthermore, one may apply this method to different initial vertices $v$ and choose the largest
clique. In the case that two cliques of the same size are found, it is advisable to follow a suggestion made by Sewell \cite{SEWELL}:
retain the clique $Q$ that maximizes $\sum_{q \in Q} d(q)$. The clique found with this criterion will lead to smaller initial sets $F(v)$ since
those colors used by the clique will not be available for vertices $v$ adjacent to some vertex in the clique. Let us call \textsc{FindClique}($G$) to this algorithm.\\

Let us notice that the distance between $\chi(G)$ and $\chi_{eq}(G)$ can be as far as we want. Such is the case with
star graphs $K_{1,m}$ \cite{MEYER} (\ie a graph $K_{1,m}$ is composed of a vertex $v$ and a stable set $S$ of size $m$ such that $v$ is adjacent to every
vertex in $S$):
\[ \chi_{eq}(K_{1,m}) - \chi(K_{1,m}) = (\lceil m/2 \rceil + 1) - 2 = \lceil m/2 \rceil - 1. \]
Therefore, it becomes essential to find other lower bounds for $\chi_{eq}(G)$ besides a maximal clique of $G$.
Lih and Chen \cite{EQTREE} proved that $$\chi_{eq}(G) \geq \biggl\lceil \dfrac{n + 1}{\alpha(V \backslash N[v]) + 2} \biggr\rceil$$ for any
$v \in V$. However, it requires to know the stability number of $G[V \backslash N[v]]$, an NP-Hard problem \cite{GAREY}.
Nevertheless, a relaxation of this value can be used instead. It is known that the cardinality of a partition in cliques of a graph is an upper bound
for the stability number of that graph. Let $PC_v$ be the cardinality of a partition in cliques of $G[V \backslash N[v]]$. The lower the size
of the partition is, the tighter the bound becomes.
Let us call \textsc{EqLowBound}($G$) to the algorithm that computes the number
\begin{equation*}
  \max \biggl\{ \biggl\lceil \dfrac{n + 1}{PC_v + 2} \biggr\rceil : v \in V \biggr\},
\end{equation*}
where $PC_v$ is obtained by the following greedy heuristic. Initially, let $G_v$ be the graph $G[V \backslash N[v]]$. We compute a maximal clique of $G_v$
and then we delete those vertices from $G_v$ that belong to the clique found. This simple procedure is repeated until $G_v$ becomes empty, and $PC_v$ is
the number of cliques found.

We want to emphasize that both procedures (\textsc{FindClique} and \textsc{EqLowBound}) could be improved,
thus obtaining better bounds of $\chi_{eq}$ but at the expense of spending more CPU time.

\subsection{Computation of upper bounds}

A known upper bound for $\chi_{eq}(G)$ is $\Delta(G) + 1$ \cite{HAJNAL}, but a slightly better one can be derived from a result stated in \cite{ORE}: ``every graph
satisfying $d(u) + d(v) \leq 2r + 1$ for every edge $(u,v)$, has a $(r+1)$-eqcol''. From this result, it is straightforward to obtain the
following relationship:
\begin{equation} \label{OREBOUND}
  \chi_{eq}(G) \leq \biggl\lceil \dfrac{\max \{ d(u) + d(v) : (u,v) \in E \} - 1}{2} \biggr\rceil + 1.
\end{equation}

Another way for finding an initial upper bound is via heuristics. In our implementation, we adopt \textsc{Naive} \cite{KUBALE} which is a heuristic
that works well and produces good solutions.
Basically, \textsc{Naive} generates a classic coloring with the algorithm $SL$ \cite{MATULA} and then re-color vertices from the biggest color
class to the smallest color class. When it is not possible, a new color is assigned to some vertex from the biggest class. The re-coloring
procedure is repeated until an equitable coloring is reached.

\subsection{Quality of the bounds}

As we said above, it is important to bear in mind the CPU time assigned to the procedures that yield the bounds and how much they
will impact in the enumerative algorithm. Since these procedures are fast heuristics, we are not sure whether they yield quality bounds.
Next, we analize them through experimentation.

This experiment and all the further ones shown in this paper were carried out on an Intel i5 CPU 750@2.67GHz with Ubuntu Linux O.S. and Intel C++ Compiler.

We denote by $LB_{FC}$ to the size of the maximal clique returned by \textsc{FindClique},
$LB_{ELB}$ to the lower bound computed by \textsc{EqLowBound},
$UB_{(\ref{OREBOUND})}$ to the upper bound given by (\ref{OREBOUND}) and
$UB_{NV}$ to the number of colors of the equitable coloring returned by \textsc{Naive}.

Random instances are generated from two parameters: the number of vertices $n$ and the probability $p$ that an edge is included in the graph.
Let us note that $p$ is approximately equal to the density of the random graph, \ie $$\dfrac{2|E|}{n(n - 1)}.$$ 

Table \ref{BOUNDTABLE} summarizes the average of the bounds over 450 ramdomly generated instances of different sizes (each row of the
table corresponds to 30 instances). Columns 1-2 show the number of vertices $n$ and probability $p$ of the evaluated instances.
Columns 3-6 display the average of $LB_{FC}$, $LB_{ELB}$, $UB_{(\ref{OREBOUND})}$ and $UB_{NV}$,
and Column 7 is the average of percentage of relative gap, \ie
$$\dfrac{100 (\min\{UB_{NV},UB_{(\ref{OREBOUND})}\} - \max\{LB_{FC}, LB_{ELB}\})}{\min\{UB_{NV},UB_{(\ref{OREBOUND})}\}}.$$

\begin{table}[h]
\begin{center}
\begin{tabular}{cc|cc|cc|c}
 & & \multicolumn{2}{|c}{Lower bound} & \multicolumn{2}{|c|}{Upper bound}\\
 $n$ & $p$ & $LB_{FC}$ & $LB_{ELB}$ & $UB_{(\ref{OREBOUND})}$ & $UB_{NV}$ & \% Rel. Gap \\
\hline
125 & 0.1 & 4.03 & 3 & 21.13 & 7.67 & 46.5 \\
125 & 0.3 & 6.2 & 5.13 & 49.9 & 20.13 & 68.37 \\
125 & 0.5 & 9.1 & 9 & 75.8 & 33.03 & 71.3 \\
125 & 0.7 & 14.13 & 17.2 & 99.3 & 46.67 & 62.53 \\
125 & 0.9 & 31 & 40.67 & 119.7 & 68.33 & 40.1 \\
\hline
250 & 0.1 & 4.23 & 3 & 38.43 & 12.27 & 64.87 \\
250 & 0.3 & 6.97 & 6 & 93.67 & 38.17 & 81.13 \\
250 & 0.5 & 10.33 & 11.03 & 146.33 & 65.77 & 82.53 \\
250 & 0.7 & 16.33 & 22.63 & 193.47 & 92.87 & 75.23 \\
250 & 0.9 & 38.33 & 63 & 236.7 & 137.17 & 53.17 \\
\hline
500 & 0.1 & 4.9 & 4 & 69.3 & 22.5 & 77.9 \\
500 & 0.3 & 7.73 & 7 & 180.2 & 72.23 & 88.83 \\
500 & 0.5 & 11.46 & 13 & 281.9 & 129.57 & 89.5 \\
500 & 0.7 & 18.6 & 28.43 & 378.67 & 184.8 & 84.1 \\
500 & 0.9 & 46.57 & 93.63 & 467.57 & 286.1 & 66.73
\end{tabular}
\end{center}
\caption{Comparison of bounds}
\label{BOUNDTABLE}
\end{table}

As we can see from Table \ref{BOUNDTABLE}, $LB_{ELB}$ is particularly useful for medium and high density graphs. The time spent in the computation of the bounds (less than a second) can be considered negligible compared to the duration of the enumerative algorithm. Therefore, it is reasonable to have on hand both lower bounds and choose the best one for each case.

Regarding $UB_{(\ref{OREBOUND})}$, it seems to be useless compared to $UB_{NV}$. Moreover, we did not find any instance such that
$UB_{NV} \geq UB_{(\ref{OREBOUND})}$ showing that \textsc{Naive} algorithm is enough to provide good upper bounds.

It is worth mentioning that medium density graphs present the worst average of relative gap. Unfortunately, this issue is transported
to the enumerative algorithm making these instances the hardest to solve.\\

We also evaluated the heuristics on a set of 64 benchmark instances, of which 60 are from a subset of DIMACS COLORLIB library \cite{DIMACS} and
the remaining 4 are Kneser graphs \cite{KNESER}. Both COLORLIB and Kneser graphs were already used by other authors for evaluating
equitable coloring algorithms (c.f. \cite{BYCBRA}).

Results are given in Tables \ref{INSTANCIAS1} and \ref{INSTANCIAS2}. Columns 1-4 show the name of the instance, its number of vertices and
edges, and its equitable chromatic number (a question mark ``?'' means $\chi_{eq}(G)$ is unknown so far). Columns 5-9 display the value
of the lower bounds, the upper bounds and the percentage of relative gap. Values marked in boldface mean they match with $\chi_{eq}(G)$.

Similarly to the previous experiment, heuristics took less than one second for almost all instances. The worst case was
\texttt{latin\_sq\_10} which took 4 seconds.

Let us note that optimality is reached in 6 instances, namely \texttt{anna}, \texttt{games120}, \texttt{homer}, \texttt{huck}, \texttt{jean}
and \texttt{le450\_25b}. \textsc{Naive} also is able to compute the optimal solution in 10 instances (\texttt{mug}*\texttt{\_}*,
*\texttt{-Insertions\_}*, \texttt{myciel4} and \texttt{kneser7\_3}). On the other hand, \textsc{FindClique} reachs the best lower bound
in 12 instances (\texttt{zeroin.i.1}, \texttt{queen7\_7}, \texttt{queen8\_12}, \texttt{mulsol.i.1}, \texttt{school1\_nsh}, \texttt{fpsol2.i.1},
\texttt{le}*\texttt{\_}* and \texttt{inithx.i.1}) while $\textsc{EqLowBound}$ reachs it only for \texttt{david}.

We conclude that heuristics presented in this section are reasonably fast, simple to implement, and suitable to provide good quality bounds
to an exact algorithm.


\section{Computational experiments} \label{SCOMPU}

In this section, we make computational experiments in order to find the best strategies for \textsc{EqDSatur} and compare it
against other exact algorithms. We work with random graphs with $n \in \{70,80\}$ and
$p \in \{0.1, 0.3, 0.5, 0.7, 0.9\}$, and with $n = 90$ and $p \in \{0.1, 0.3, 0.9\}$.
For each combination of $n$ and $p$, we generate $T = 30$ instances
and we analize the performance of our algorithm by considering the following indicators:
\begin{itemize}
\item \emph{Percentage of solved instances} (\% solved): An instance is considered ``solved'' when the time
needed to reach the optimal value is at most 2 hours. The percentage of solved instances is the
value $100.|S|/T$ where $S$ is the set of solved instances.
\item \emph{Average of the best upper bound reached} (Av. UB): It is the average of the upper bound obtained after
the enumeration, over all $T$ instances.
\item \emph{Average of nodes evaluated} (Av. Nodes): It is the average of nodes evaluated of the search tree
over the set of solved instances $S$. 
\item \emph{Average of time elapsed} (Av. Time): It is the average of time in seconds needed to solve each instance,
over the set of solved instances $S$.
\end{itemize}
We report them on tables, where each row corresponds to a different combination of $n$ and $p$, and each
column displays the value of an indicator for the strategy to be compared. In general, best values are marked in boldface.
We do not evaluate combinations $n = 90$ with $p \in \{0.5, 0.7\}$ since DSatur-based algorithms (including ours)
solves few instances in those cases and comparisons become rough. The total number of instances amounts to 390.\\

When we compare two strategies $A$ and $B$, it may happen that the instances solved by $A$ and $B$ are different and the
comparison of the averages of nodes and time may be ambiguous or unfair. In those cases, we consider these averages over
the set of instances solved by both strategies: if $S_A$ and $S_B$ are the set of solved instances for $A$ and $B$ respectively,
we also compute the average of nodes and time over the set $S_A \cap S_B$. These values are reported with a mark ``$\dagger$''.

\subsection{Vertex selection strategy}

The following experiment compares an implementation of \textsc{EqDSatur} with the three vertex selection strategies mentioned in Section \ref{SDSATUR}
namely DSATUR-VSS, CELIM-VSS and PASS-VSS.
Tables \ref{TABLEVSS1}-\ref{TABLEVSS2} resume the results.

\begin{table}[!h]
\begin{center} \footnotesize
\begin{tabular}{c|c|c@{\hspace{4pt}}c@{\hspace{4pt}}c|c@{\hspace{4pt}}c@{\hspace{4pt}}c}
    &   & \multicolumn{3}{|c}{\% solved} & \multicolumn{3}{|c}{Av. UB} \\
$n$ & $p$ & DSATUR & CELIM & PASS & DSATUR & CELIM & PASS \\
\hline
70 & 0.1 & 100 & 100 & 100 & 4 & 4 & 4 \\
70 & 0.3 & 100 & 100 & 100 & 7.93 & 7.93 & 7.93 \\
70 & 0.5 & 93 & 97 & \textbf{100} & 12.03 & 11.93 & \textbf{11.83} \\
70 & 0.7 & 97 & \textbf{100} & \textbf{100} & 17.53 & \textbf{17.3} & \textbf{17.3} \\
70 & 0.9 & 100 & 100 & 100 & 29.2 & 29.2 & 29.2 \\
\hline
80 & 0.1 & 100 & 100 & 100 & 4.23 & 4.23 & 4.23 \\
80 & 0.3 & \textbf{100} & 97 & \textbf{100} & \textbf{8.43} & 8.53 & \textbf{8.43} \\
80 & 0.5 & 87 & 87 & \textbf{93} & 13.47 & 13.47 & \textbf{13.2} \\
80 & 0.7 & 53 & 50 & \textbf{70} & 20.1 & 20.2 & \textbf{19.53} \\
80 & 0.9 & 100 & 100 & 100 & 31.7 & 31.7 & 31.7 \\
\hline
90 & 0.1 & 100 & 100 & 100 & 5 & 5 & 5 \\
90 & 0.3 & 100 & 100 & 100 & 9 & 9 & 9 \\
90 & 0.9 & 100 & 100 & 100 & 34.2 & 34.2 & 34.2
\end{tabular}
\end{center}
\caption{Tests on different vertex selection strategies}
\label{TABLEVSS1}
\end{table}

\begin{table}[!h]
\begin{center} \footnotesize
\begin{tabular}{c|c|c@{\hspace{4pt}}c@{\hspace{4pt}}c|c@{\hspace{4pt}}c@{\hspace{4pt}}c}
    &   & \multicolumn{3}{|c}{Av. Nodes} & \multicolumn{3}{|c}{Av. Time} \\
$n$ & $p$ & DSATUR & CELIM & PASS & DSATUR & CELIM & PASS \\
\hline
70 & 0.1 & 216 & \textbf{168} & 208 & 0 & 0 & 0 \\
70 & 0.3 & 401862 & 253181 & \textbf{171448} & 0.1 & 0.1 & 0.07 \\
70 & 0.5 & 6116237 & 5134138 & \textbf{4702843} & 6.61 & 7.76 & 6.7 \\
70 & 0.7 & 21048794 & \textbf{11175213} & 12020710 & 28.2 & \textbf{23.5} & \textbf{21} \\
70 & 0.9 & 249682 & 145481 & \textbf{138057} & 0.17 & 0.17 & 0.1 \\
\hline
80 & 0.1 & 1132 & \textbf{967} & 5186 & 0 & 0 & 0 \\
80 & 0.3 & 31102992 & 17153530 & \textbf{15495305} & 27.7 & 25.2 & \textbf{22.7} \\
80 & 0.5 & 540416906 & 333631281 & \textbf{192172556} & 601 & 574 & \textbf{324} \\
80 & 0.7 & 821110267 & 480890653 & 959670395 & 1263 & 1162 & 1817 \\
 &  & 675165908$^\dagger$ & \textbf{308409086}$^\dagger$ & 410011950$^\dagger$ & 1035$^\dagger$ & \textbf{749}$^\dagger$ & 791$^\dagger$ \\
80 & 0.9 & 5513947 & \textbf{3098276} & 3817790 & 8.2 & 7.6 & 6.57 \\
\hline
90 & 0.1 & 4521 & 3186 & \textbf{2875} & 0 & 0 & 0 \\
90 & 0.3 & 83857234 & 58179096 & \textbf{32510740} & 86.8 & 88.6 & \textbf{52} \\
90 & 0.9 & 144093673 & \textbf{71388770} & 73185398 & 305 & 218 & \textbf{161}
\end{tabular}
\end{center}
\caption{Tests on different vertex selection strategies}
\label{TABLEVSS2}
\end{table}

As we can see, PASS-VSS has been able to solve more instances than the other strategies.
Also, PASS-VSS performs better in terms of time. Nevertheless, DSATUR-VSS and CELIM-VSS reports less time than PASS-VSS for graphs of 80
vertices and $p = 0.7$. Since PASS-VSS has solved more instances than the other two strategies, we have added an
extra row marked with ``$\dagger$'' reporting averages for the three strategies over the instances that the three strategies have been able to
solve simultaneously. Here, CELIM-VSS seems to be a little better than PASS-VSS.
In our opinion, it is not worth considering these small improvements at the expense of solving fewer instances.

Our conclusion is that PASS-VSS is the right choice for our algorithm.

\subsection{Color selection strategy}

We contemplate four options:
\begin{itemize}
\item \emph{DSATUR-CSS}. Consider the set of available colors in ascending order.
\item \emph{BCCOL-CSS} \cite{BCCOLBRANCHINGRULE}. First consider the new color ($k+1$) and then the set of
available colors in ascending order.
\item \emph{ORDER1-CSS}. Sort color classes of $\Pi$ according to their size in ascending order: $|C_{i_1}| \leq |C_{i_2}| \leq \ldots \leq |C_{i_k}|$.
Then consider colors in the following order: $i_1$, $i_2$, $\ldots$, $i_k$, $k+1$.
\item \emph{ORDER2-CSS}. Do the same as in ORDER1-CSS but considering colors in the following order: $k+1$, $i_1$, $i_2$, $\ldots$, $i_k$.
\end{itemize}

BCCOL-CSS is implemented as part of the branching strategy in the Branch-and-Cut \textsc{BC-Col} and the idea is that it tends to
find feasible colorings quickly, albeit not good since it introduces new colors to reach them.
ORDER1-CSS is inspired in the heuristic presented in \cite{BRELAZEQUIT}.
This rule tends to balance the sizes of color classes and finds equitable colorings early.
The downside is that a QuickSort must be performed on each node.
ORDER2-CSS is a mix between ORDER1-CSS and BCCOL-CSS. Since we have noticed that it does not perform as well as the others, we do not
report it.

Results for DSATUR-CSS, BCCOL-CSS and ORDER1-CSS are resumed in Tables \ref{TABLECSS1}-\ref{TABLECSS2}.

\begin{table}[!h]
\begin{center} \footnotesize
\begin{tabular}{c|c|c@{\hspace{4pt}}c@{\hspace{4pt}}c|c@{\hspace{4pt}}c@{\hspace{4pt}}c}
    &   & \multicolumn{3}{|c}{\% solved} & \multicolumn{3}{|c}{Av. UB} \\
$n$ & $p$ & DSATUR & BCCOL & ORDER1 & DSATUR & BCCOL & ORDER1 \\
\hline
70 & 0.1 & 100 & 100 & 100 & 4 & 4 & 4 \\
70 & 0.3 & 100 & 100 & 100 & 7.93 & 7.93 & 7.93 \\
70 & 0.5 & 100 & 100 & 100 & 11.8 & 11.8 & 11.8 \\
70 & 0.7 & \textbf{100} & 93 & \textbf{100} & \textbf{17.3} & 17.8 & \textbf{17.3} \\
70 & 0.9 & \textbf{100} & 77 & \textbf{100} & \textbf{29.2} & 30.8 & \textbf{29.2} \\
\hline
80 & 0.1 & 100 & 100 & 100 & 4.23 & 4.23 & 4.23 \\
80 & 0.3 & \textbf{100} & 93 & \textbf{100} & \textbf{8.43} & 8.63 & \textbf{8.43} \\
80 & 0.5 & 93 & 93 & 93 & \textbf{13.2} & 13.3 & \textbf{13.2} \\
80 & 0.7 & 70 & \textbf{73} & 70 & \textbf{19.5} & 20.3 & \textbf{19.5} \\
80 & 0.9 & \textbf{100} & 90 & \textbf{100} & \textbf{31.7} & 32.5 & \textbf{31.7} \\
\hline
90 & 0.1 & 100 & 100 & 100 & 5 & 5 & 5 \\
90 & 0.3 & 100 & 100 & 100 & 9 & 9 & 9 \\
90 & 0.9 & \textbf{100} & 80 & \textbf{100} & \textbf{34.2} & 36.5 & \textbf{34.2}
\end{tabular}
\end{center}
\caption{Tests on different color section strategies}
\label{TABLECSS1}
\end{table}

\begin{table}[!h]
\begin{center} \footnotesize
\begin{tabular}{c|c|c@{\hspace{4pt}}c@{\hspace{4pt}}c|c@{\hspace{4pt}}c@{\hspace{4pt}}c}
    &   & \multicolumn{3}{|c}{Av. Nodes} & \multicolumn{3}{|c}{Av. Time} \\
$n$ & $p$ & DSATUR & BCCOL & ORDER1 & DSATUR & BCCOL & ORDER1 \\
\hline
70 & 0.1 & \textbf{208} & \textbf{208} & 311 & 0 & 0 & 0 \\
70 & 0.3 & 171448 & \textbf{139351} & 166733 & 0.07 & 0.03 & 0.07 \\
70 & 0.5 & \textbf{4702843} & 34141331 & 10586393 & \textbf{6.7} & 35.9 & 14.3 \\
70 & 0.7 & \textbf{12020710} & 116668568 & \textbf{12120596} & \textbf{21} & 109 & 25.1 \\
70 & 0.9 & \textbf{138057} & 11987145 & \textbf{138058} & \textbf{0.1} & 10.3 & \textbf{0.13} \\
\hline
80 & 0.1 & 5186 & \textbf{932} & 1006 & 0 & 0 & 0 \\
80 & 0.3 & \textbf{15495305} & 17812892 & \textbf{15394052} & 22.7 & 21.9 & 23 \\
80 & 0.5 & 192172555 & 270041809 & \textbf{179275549} & 324 & 379 & \textbf{318} \\
80 & 0.7 & 959670395 & 810826785 & 941362879 & 1817 & 1179 & 1807 \\
 & & 1052136994$^\dagger$ & \textbf{923999573}$^\dagger$ & 1028274591$^\dagger$ & 1978$^\dagger$ & \textbf{1259}$^\dagger$ & 1963$^\dagger$ \\
80 & 0.9 & \textbf{3817790} & 42009653 & \textbf{3818024} & \textbf{6.57} & 37.7 & \textbf{6.8} \\
\hline
90 & 0.1 & \textbf{2875} & 5907 & 2909 & 0 & 0 & 0 \\
90 & 0.3 & \textbf{32510740} & 47623951 & 44847604 & \textbf{52} & 74.1 & 65.6 \\
90 & 0.9 & \textbf{73185398} & 164689901 & \textbf{73184947} & \textbf{161} & 253 & 168
\end{tabular}
\end{center}
\caption{Tests on different color section strategies}
\label{TABLECSS2}
\end{table}

We first analyze the differences between the classical strategy DSATUR-CSS and BCCOL-CSS, where the latter performs quite well for $n = 80$
and $p = 0.7$. We have noticed that both strategies do not solve the same instances, hence the discrepancy between solved
instances (70\% and 73\% respectively) and average of $UB$ (19.5 and 20.3 respectively), so we have added an extra row reporting
averages over the instances that both strategies have been able to solve simultaneously. Although, by inspecting the extra row, BCCOL-CSS
solves the ``common'' instances 57\% faster than DSATUR-CSS, the performance of BCCOL-CSS is worse for most of the remaining rows.

Regarding ORDER1-CSS, we can note that there are few differences between this strategy and DSATUR-CSS. Both strategies solves the same
instances and reaches the same $UB$ for every non-solved graph. The time used by DSATUR-CSS is slightly less than ORDER1-CSS for graphs of 70 and 90
vertices. For $n = 80$ and $p \in \{ 30, 50\}$, ORDER1-CSS evaluates 7\% and 2\% less nodes respectively than DSATUR-CSS. Since ORDER1-CSS
performs a QuickSort at each node, the differences in time among these strategies fall to 2\% and 0.6\% respectively.

We choose DSATUR-CSS for our implementation of \textsc{EqDSatur}, but ORDER1-CSS may be considered as an alternative strategy anywise.

\subsection{TrivialEqDSatur vs. EqDSatur}

Our next experiment consists of comparing \textsc{TrivialEqDSatur} and \textsc{EqDSatur} implementations
in order to verify whether the pruning rule given in Section \ref{SNEWPRUN} is efficient. We recall that \textsc{TrivialEqDSatur} is a simple modification
of the standard \textsc{DSatur} that checks whether the colorings at the leafs of the search tree are equitable or not. Both
algorithms use the same selection strategies previously chosen and the same bounds given by the heuristics proposed in Section
\ref{SBOUNDS} (although \textsc{TrivialEqDSatur} does not take advantage of the value of $LB$).
Table \ref{TRIVIALVSNOTRIVIAL} resumes the results.

\begin{table}[!h]
\begin{center} \footnotesize
\begin{tabular}{c|c|c@{\hspace{4pt}}c|c@{\hspace{4pt}}c|c@{\hspace{4pt}}c|c@{\hspace{4pt}}c}
  &  & \multicolumn{2}{|c}{\% solved} & \multicolumn{2}{|c}{Av. UB} & \multicolumn{2}{|c}{Av. Nodes} & \multicolumn{2}{|c}{Av. Time} \\
$n$ & $p$ & Triv. & EqDS & Triv. & EqDS & Triv. & EqDS & Triv. & EqDS \\
\hline
70 & 0.1 & 100 & 100 & 4 & 4 & 264721 & \textbf{208} & 0 & 0 \\
70 & 0.3 & 100 & 100 & 7.93 & 7.93 & 168862113 & \textbf{171448} & 32.7 & \textbf{0.07} \\
70 & 0.5 & 100 & 100 & 11.8 & 11.8 & 88287477 & \textbf{4702843} & 29.9 & \textbf{6.7} \\
70 & 0.7 & 100 & 100 & 17.3 & 17.3 & 37918448 & \textbf{12020710} & 29 & \textbf{21} \\
70 & 0.9 & 100 & 100 & 29.2 & 29.2 & 2776802 & \textbf{138057} & 1.07 & 0.1 \\
\hline
80 & 0.1 & 100 & 100 & 4.23 & 4.23 & 130316183 & \textbf{5186} & 20.6 & \textbf{0} \\
80 & 0.3 & 100 & 100 & 8.43 & 8.43 & 345842251 & \textbf{15495305} & 99.6 & \textbf{22.7} \\
80 & 0.5 & 83 & \textbf{93} & 13.6 & \textbf{13.2} & 1614284274 & \textbf{192172556} & 705 & \textbf{324} \\
80 & 0.7 & 67 & \textbf{70} & 19.6 & \textbf{19.5} & 897961603 & 959670395 & 1665 & 1817 \\
   &    &     &             &      &               & 897961603$^\dagger$ & \textbf{828204878}$^\dagger$ & 1665$^\dagger$ & \textbf{1585}$^\dagger$\\
80 & 0.9 & 100 & 100 & 31.7 & 31.7 & 122216644 & \textbf{3817790} & 54 & \textbf{6.57} \\
\hline
90 & 0.1 & 100 & 100 & 5 & 5 & 15428656 & \textbf{2875} & 2.47 & \textbf{0} \\
90 & 0.3 & 100 & 100 & 9 & 9 & 124572212 & \textbf{32510740} & 75.9 & \textbf{52} \\
90 & 0.9 & 100 & 100 & 34.2 & 34.2 & 75124470 & \textbf{73185398} & 169 & \textbf{161}
\end{tabular}
\end{center}
\caption{Comparison between \textsc{TrivialEqDSatur} and \textsc{EqDSatur}}
\label{TRIVIALVSNOTRIVIAL}
\end{table}

We have noticed that every instance solved by \textsc{TrivialEqDSatur} has been solved by \textsc{EqDSatur}
too, but not conversely. This fact led us to insert an extra row in the table for the case $n = 80$ and $p = 0.7$, where we
report the average of nodes evaluated and time elapsed of \textsc{EqDSatur} for those instances that have been solved by \textsc{TrivialEqDSatur}.

We can observe that \textsc{EqDSatur} outperforms \textsc{TrivialEqDSatur} for all the indicators.

\subsection{Comparing against other exact algorithms}

This subsection is devoted to compare \textsc{EqDSatur} against the Branch-and-Cut B\&C-$LF_2$ described in \cite{BYCBRA} and the general purpose
solver CPLEX 12.4 with the IP formulation given in \cite{PAPERDAM} and the initial bounds computed by the heuristics given in Section
\ref{SBOUNDS}.

In the first experiment, we consider 30 instances for each combination of $n \in \{60, 70\}$ and $p \in \{0.1, 0.3, 0.5, 0.7, 0.9\}$.
We also consider $n \in \{ 80, 100, 120 \}$ with $p = 0.1$ and $n = 80$ with $p = 0.9$ since CPLEX solves very few medium-density
random instances with $n \geq 80$. The total number of instances amounts to 420.

Table \ref{RANDOMFINAL} summarizes the results, where $LB$ and $UB$ are averaged over all instances while the time elapsed is averaged
over solved instances. A mark ``$-$'' is reported when no instance is solved. Columns called ``Init.'' correspond to the bounds computed by the initial heuristics.

\begin{table}[!h]
\begin{center} \small
\begin{tabular}{c|c|c@{\hspace{4pt}}c|c@{\hspace{4pt}}c|c@{\hspace{4pt}}c@{\hspace{4pt}}c|c@{\hspace{4pt}}c}
  &  & \multicolumn{2}{|c}{\% solved} & \multicolumn{2}{|c}{Av. LB} & \multicolumn{3}{|c}{Av. UB} & \multicolumn{2}{|c}{Av. Time} \\
$n$ & $p$ & CPX & EqDS & Init. & CPX & Init. & CPX & EqDS & CPX & EqDS \\
\hline
60 & 0.1 & 100 & 100 & 3.23 & 4 & 4.9 & 4 & 4 & 0 & 0 \\
60 & 0.3 & 100 & 100 & 5.23 & 7.03 & 10.7 & 7.03 & 7.03 & 506 & 0 \\
60 & 0.5 & 63 & 100 & 7.63 & 10.6 & 17.1 & 11 & 10.8 & 1825 & 0.6 \\
60 & 0.7 & 90 & 100 & 12.7 & 15.7 & 22.7 & 15.8 & 15.7 & 942 & 1.6 \\
60 & 0.9 & 100 & 100 & 22.8 & 26 & 32.3 & 26 & 26 & 1 & 0 \\
\hline
70 & 0.1 & 100 & 100 & 3.5 & 4 & 5.03 & 4 & 4 & 0.03 & 0 \\
70 & 0.3 & 50 & 100 & 5.33 & 7.5 & 12.6 & 8 & 7.93 & 4005 & 0.07 \\
70 & 0.5 & 0 & 100 & 8.13 & 11 & 19.1 & 12.8 & 11.8 & $-$ & 6.7 \\
70 & 0.7 & 20 & 100 & 13.9 & 16.7 & 26.6 & 18.2 & 17.3 & 2360 & 21 \\
70 & 0.9 & 100 & 100 & 25.1 & 29.2 & 37.7 & 29.2 & 29.2 & 258 & 0.1  \\
\hline
80 & 0.1 & 100 & 100 & 3.67 & 4.23 & 5.63 & 4.23 & 4.23 & 1.33 & 0 \\
80 & 0.9 & 90 & 100 & 27.1 & 31.6 & 42.1 & 31.7 & 31.7 & 659 & 6.57 \\
100 & 0.1 & 100 & 100 & 3.9 & 5 & 6.87 & 5 & 5 & 15.5 & 0 \\
120 & 0.1 & 50 & 100 & 4 & 5 & 7.73 & 5.5 & 5.1 & 1673 & 2.57
\end{tabular}
\end{center}
\caption{Performance of \textsc{EqDSatur} and CPLEX on random graphs}
\label{RANDOMFINAL}
\end{table}

We note that our algorithm is able to solve more instances than CPLEX in considerably less time. The differences are more pronounced in
medium density instances.

We do not compare \textsc{EqDSatur} directly against B\&C-$LF_2$ since values reported in \cite{BYCBRA} consider
different random instances. Despite this, we remark that B\&C-$LF_2$ has failed to solve any instance
with $n = 70$ and $p \in \{0.3, 0.5\}$ whereas \textsc{EqDSatur} can solve instances of the same
size without difficulty.\\

The last experiment consists of comparing \textsc{EqDSatur} against CPLEX and B\&C-$LF_2$ on DIMACS COLORLIB instances and
Kneser graphs proposed in Section \ref{SBOUNDS}, except those instances that have been already solved by the initial heuristics.
Besides DSATUR-CSS, we also take into account the alternative color strategy ORDER1-CSS.
\begin{table}[ht]
\begin{center}
\begin{tabular}{cccc|cc|cc|c}
 & & & & \multicolumn{2}{|c}{Lower Bound} & \multicolumn{2}{|c|}{Upper Bound} & \\
 {\footnotesize Name} & {\footnotesize Vert.} & {\footnotesize Edges} & {\footnotesize $\chi_{eq}$} & {\footnotesize $LB_{FC}$} &
 {\footnotesize $LB_{ELB}$} & {\footnotesize $UB_{(\ref{OREBOUND})}$} & {\footnotesize $UB_{NV}$} & {\footnotesize \% Rel. Gap.} \\
\hline
miles750 & 128 & 2113 & 31 & 30 & 11 & 64 & 33 & 9.09 \\
miles1000 & 128 & 3216 & 42 & 40 & 17 & 87 & 47 & 14.89 \\
miles1500 & 128 & 5198 & 73 & 69 & 43 & 107 & 74 & 6.76 \\
zeroin.i.1 & 211 & 4100 & 49 & \textbf{49} & 3 & 111 & 51 & 3.92 \\
zeroin.i.2 & 211 & 3541 & 36 & 30 & 4 & 141 & 51 & 41.18 \\
zeroin.i.3 & 206 & 3540 & 36 & 30 & 4 & 141 & 49 & 38.78 \\
queen6\_6 & 36 & 290 & 7 & 6 & 5 & 20 & 10 & 40 \\
queen7\_7 & 49 & 476 & 7 & \textbf{7} & 6 & 24 & 12 & 41.67 \\
queen8\_8 & 64 & 728 & 9 & 8 & 8 & 28 & 18 & 55.56 \\
queen8\_12 & 96 & 1368 & 12 & \textbf{12} & 11 & 33 & 20 & 40 \\
queen9\_9 & 81 & 1056 & 10 & 9 & 8 & 32 & 15 & 40 \\
queen10\_10 & 100 & 1470 & ? & 10 & 10 & 36 & 18 & 44.44 \\
anna & 138 & 493 & 11 & \textbf{11} & 3 & 61 & \textbf{11} & 0 \\
david & 87 & 406 & 30 & 11 & \textbf{30} & 59 & 40 & 25 \\
games120 & 120 & 638 & 9 & \textbf{9} & 5 & 14 & \textbf{9} & 0 \\
homer & 561 & 1628 & 13 & \textbf{13} & 2 & 89 & \textbf{13} & 0 \\
huck & 74 & 301 & 11 & \textbf{11} & 6 & 40 & \textbf{11} & 0 \\
jean & 80 & 254 & 10 & \textbf{10} & 3 & 30 & \textbf{10} & 0 \\
1-FullIns\_3 & 30 & 100 & 4 & 3 & 3 & 12 & 7 & 57.14 \\
2-FullIns\_3 & 52 & 201 & 5 & 4 & 3 & 16 & 9 & 55.56 \\
3-FullIns\_3 & 80 & 346 & 6 & 5 & 3 & 20 & 7 & 28.57 \\
4-FullIns\_3 & 114 & 541 & 7 & 6 & 3 & 24 & 12 & 50 \\
5-FullIns\_3 & 154 & 792 & 8 & 7 & 3 & 28 & 9 & 22.22 \\
1-FullIns\_4 & 93 & 593 & 5 & 3 & 3 & 33 & 7 & 57.14 \\
mug88\_1 & 88 & 146 & 4 & 3 & 3 & 5 & \textbf{4} & 25 \\
mug88\_25 & 88 & 146 & 4 & 3 & 3 & 5 & \textbf{4} & 25 \\
mug100\_1 & 100 & 166 & 4 & 3 & 3 & 5 & \textbf{4} & 25 \\
mug100\_25 & 100 & 166 & 4 & 3 & 3 & 5 & \textbf{4} & 25 \\
mulsol.i.1 & 197 & 3925 & 49 & \textbf{49} & 4 & 122 & 63 & 22.22 \\
mulsol.i.2 & 188 & 3885 & ? & 31 & 11 & 157 & 58 & 46.55 \\
school1 & 385 & 19095 & 15 & 14 & 9 & 278 & 49 & 71.43 \\
school1\_nsh & 352 & 14612 & 14 & \textbf{14} & 8 & 231 & 40 & 65
\end{tabular}
\end{center}
\caption{COLORLIB instances (part 1)}
\label{INSTANCIAS1}
\end{table}
\begin{table}[h]
\begin{center}
\begin{tabular}{cccc|cc|cc|c}
 & & & & \multicolumn{2}{|c}{Lower Bound} & \multicolumn{2}{|c|}{Upper Bound} & \\
 {\footnotesize Name} & {\footnotesize Vert.} & {\footnotesize Edges} & {\footnotesize $\chi_{eq}$} & {\footnotesize $LB_{FC}$} &
 {\footnotesize $LB_{ELB}$} & {\footnotesize $UB_{(\ref{OREBOUND})}$} & {\footnotesize $UB_{NV}$} & {\footnotesize \% Rel. Gap.}\\
\hline
fpsol2.i.1 & 496 & 11654 & 65 & \textbf{65} & 3 & 253 & 85 & 23.53 \\
fpsol2.i.2 & 451 & 8691 & 47 & 30 & 5 & 347 & 62 & 51.61 \\
fpsol2.i.3 & 425 & 8688 & 55 & 30 & 7 & 347 & 80 & 62.5 \\
1-Insertions\_4 & 67 & 232 & 5 & 2 & 3 & 16 & \textbf{5} & 40 \\
2-Insertions\_3 & 37 & 72 & 4 & 2 & 3 & 7 & \textbf{4} & 25 \\
3-Insertions\_3 & 56 & 110 & 4 & 2 & 3 & 8 & \textbf{4} & 25 \\
4-Insertions\_3 & 79 & 156 & 4 & 2 & 2 & 9 & \textbf{4} & 50 \\
DSJC125.1 & 125 & 736 & 5 & 4 & 3 & 22 & 8 & 50 \\
DSJC125.5 & 125 & 3891 & ? & 9 & 9 & 75 & 27 & 66.67 \\
DSJC125.9 & 125 & 6961 & ? & 30 & 42 & 120 & 66 & 36.36 \\
DSJC250.1 & 250 & 3218 & ? & 4 & 3 & 37 & 13 & 69.23 \\
DSJC250.5 & 250 & 15668 & ? & 10 & 11 & 146 & 65 & 83.08 \\
DSJC250.9 & 250 & 27897 & ? & 37 & 63 & 235 & 136 & 53.68 \\
le450\_5a & 450 & 5714 & 5 & \textbf{5} & 3 & 41 & 12 & 58.33 \\
le450\_5b & 450 & 5734 & 5 & \textbf{5} & 4 & 41 & 12 & 58.33 \\
le450\_15a & 450 & 8168 & 15 & \textbf{15} & 5 & 89 & 18 & 16.67 \\
le450\_15b & 450 & 8169 & 15 & \textbf{15} & 5 & 91 & 17 & 11.76 \\
le450\_25a & 450 & 8260 & 25 & \textbf{25} & 5 & 118 & 26 & 3.85 \\
le450\_25b & 450 & 8263 & 25 & \textbf{25} & 6 & 107 & \textbf{25} & 0 \\
inithx.i.1 & 864 & 18707 & 54 & \textbf{54} & 3 & 503 & 70 & 22.86 \\
inithx.i.2 & 645 & 13979 & ? & 30 & 8 & 542 & 158 & 81.01 \\
myciel4 & 23 & 71 & 5 & 2 & 3 & 9 & \textbf{5} & 40 \\
myciel5 & 47 & 236 & 6 & 2 & 3 & 18 & 9 & 66.67 \\
myciel6 & 95 & 755 & ? & 2 & 3 & 36 & 11 & 72.73 \\
flat300\_20\_0 & 300 & 21375 & ? & 10 & 11 & 160 & 81 & 86.42 \\
latin\_sq\_10 & 900 & 307350 & ? & 90 & 82 & 684 & 460 & 80.43 \\
ash331GPIA & 662 & 4181 & 4 & 3 & 3 & 24 & 8 & 62.5 \\
will199GPIA & 701 & 6772 & 7 & 6 & 4 & 39 & 9 & 33.33 \\
kneser7\_2 & 21 & 105 & 6 & 3 & 3 & 11 & 8 & 62.5 \\
kneser7\_3 & 35 & 70 & 3 & 2 & 2 & 5 & \textbf{3} & 33.33 \\
kneser9\_4 & 126 & 315 & 3 & 2 & 2 & 6 & 4 & 50 \\
kneser11\_5 & 462 & 1386 & 3 & 2 & 2 & 7 & 4 & 50
\end{tabular}
\end{center}
\caption{COLORLIB instances (part 2) and Kneser graphs}
\label{INSTANCIAS2}
\end{table}
\begin{table}[h]
\begin{center} 
\tiny
\begin{tabular}{c|c|c@{\hspace{4pt}}c@{\hspace{4pt}}c|c@{\hspace{4pt}}c@{\hspace{4pt}}c@{\hspace{4pt}}c@{\hspace{4pt}}c|c@{\hspace{4pt}}c@{\hspace{4pt}}c@{\hspace{4pt}}c}
 & & \multicolumn{3}{|c}{Lower Bound} & \multicolumn{5}{|c}{Upper Bound} & \multicolumn{4}{|c}{Time} \\
Name & $\chi_{eq}$ & Init. & CPX & BC$LF_2$ & Init. & CPX & BC$LF_2$ & EqDS & EqDS$^*$ & CPX & BC$LF_2$ & EqDS & EqDS$^*$ \\ 
\hline
miles750 & 31 & 30 & 31 & 31 & 33 & 31 & 31 & 33 & 31 & 0 & 171 & $-$ & 0 \\ 
miles1000 & 42 & 40 & 42 & 42 & 47 & 42 & 42 & 47 & 42 & 0 & 267 & $-$ & 0 \\ 
miles1500 & 73 & 69 & 73 & 73 & 74 & 73 & 73 & 73 & 73 & 0 & 13 & 0 & 0 \\ 
zeroin.i.1 & 49 & 49 & 49 & 49 & 51 & 49 & 49 & 49 & 49 & 0 & 50 & 0 & 0 \\ 
zeroin.i.2 & 36 & 30 & 36 & 36 & 51 & 36 & 36 & 51 & 51 & 2 & 510 & $-$ & $-$ \\ 
zeroin.i.3 & 36 & 30 & 36 & 36 & 49 & 36 & 36 & 49 & 49 & 5 & 491 & $-$ & $-$ \\
queen6\_6 & 7 & 6 & 7 & 7 & 10 & 7 & 7 & 7 & 7 & 1 & 1 & 0 & 0 \\ 
queen7\_7 & 7 & 7 & 7 & 7 & 12 & 7 & 7 & 7 & 7 & 0 & 0 & 0 & 0 \\ 
queen8\_8 & 9 & 8 & 9 & 9 & 18 & 9 & 9 & 9 & 9 & 654 & 441 & 6 & 1 \\ 
queen8\_12 & 12 & 12 & 12 & & 20 & 12 & & 12 & 20 & 5 & & 3079 & $-$ \\ 
queen9\_9 & 10 & 9 & 9 & & 15 & 11 & & 10 & 10 & $-$ & & 475 & 499 \\ 
queen10\_10 & ? & 10 & 10 & & 18 & 12 & & 13 & 11 & $-$ & & $-$ & $-$ \\ 
david & 30 & 30 & 30 & 30 & 40 & 30 & 30 & 30 & 30 & 0 & 13 & 0 & 0 \\ 
1-FullIns\_3 & 4 & 3 & 4 & 4 & 7 & 4 & 4 & 4 & 4 & 0 & 2 & 0 & 0 \\ 
2-FullIns\_3 & 5 & 4 & 5 & 5 & 9 & 5 & 5 & 5 & 5 & 0 & 25 & 1 & 1 \\ 
3-FullIns\_3 & 6 & 5 & 6 & 6 & 7 & 6 & 6 & 7 & 7 & 0 & 85 & $-$ & $-$ \\ 
4-FullIns\_3 & 7 & 6 & 7 & 7 & 12 & 7 & 7 & 12 & 7 & 0 & 72 & $-$ & $-$ \\ 
5-FullIns\_3 & 8 & 7 & 8 & 8 & 9 & 8 & 8 & 9 & 9 & 0 & 268 & $-$ & $-$ \\ 
1-FullIns\_4 & 5 & 3 & 5 & & 7 & 5 & & 5 & 5 & 28 & & 1404 & 1412 \\ 
mug88\_1 & 4 & 3 & 4 & & 4 & 4 & & 4 & 4 & 1 & & 109 & 120 \\ 
mug88\_25 & 4 & 3 & 4 & & 4 & 4 & & 4 & 4 & 0 & & 56 & 60 \\ 
mug100\_1 & 4 & 3 & 4 & & 4 & 4 & & 4 & 4 & 1 & & 4425 & 4946 \\ 
mug100\_25 & 4 & 3 & 4 & & 4 & 4 & & 4 & 4 & 1 & & 4978 & 5595 \\ 
mulsol.i.1 & 49 & 49 & 49 & & 63 & 49 & & 49 & 49 & 1 & & 0 & 0 \\ 
mulsol.i.2 & ? & 31 & 34 & & 58 & 39 & & 58 & 58 & $-$ & & $-$ & $-$ \\ 
school1 & 15 & 14 & 14 & & 49 & 49 & & 49 & 49 & $-$ & & $-$ & $-$ \\ 
school1\_nsh & 14 & 14 & 14 & & 40 & 14 & & 23 & 40 & 1840 & & $-$ & $-$ \\ 
fpsol2.i.1 & 65 & 65 & 65 & & 85 & 65 & & 65 & 65 & 11 & & 0 & 0 \\ 
fpsol2.i.2 & 47 & 30 & 47 & & 62 & 62 & & 62 & 62 & $-$ & & $-$ & $-$ \\ 
fpsol2.i.3 & 55 & 30 & 55 & & 80 & 80 & & 80 & 80 & $-$ & & $-$ & $-$ \\ 
1-Insertions\_4 & 5 & 3 & 4 & & 5 & 5 & & 5 & 5 & $-$ & & 1055 & 1088 \\ 
2-Insertions\_3 & 4 & 3 & 4 & & 4 & 4 & & 4 & 4 & 0 & & 0 & 0 \\ 
3-Insertions\_3 & 4 & 3 & 4 & & 4 & 4 & & 4 & 4 & 8 & & 1 & 2 \\ 
4-Insertions\_3 & 4 & 2 & 4 & & 4 & 4 & & 4 & 4 & 836 & & 1615 & 1701 \\ 
DSJC125.1 & 5 & 4 & 5 & & 8 & 5 & & 5 & 5 & 214 & & 0 & 0 \\ 
DSJC125.5 & ? & 9 & 13 & & 27 & 27 & & 19 & 19 & $-$ & & $-$ & $-$ \\ 
DSJC125.9 & ? & 42 & 43 & & 66 & 47 & & 47 & 47 & $-$ & & $-$ & $-$ \\ 
DSJC250.1 & ? & 4 & 5 & & 13 & 13 & & 9 & 9 & $-$ & & $-$ & $-$ \\ 
DSJC250.5 & ? & 11 & 12 & & 65 & 65 & & 36 & 65 & $-$ & & $-$ & $-$ \\ 
DSJC250.9 & ? & 63 & 63 & & 136 & 136 & & 86 & 86 & $-$ & & $-$ & $-$ \\ 
le450\_5a & 5 & 5 & 5 & & 12 & 5 & & 12 & 10 & 4558 & & $-$ & $-$ \\ 
le450\_5b & 5 & 5 & 5 & & 12 & 5 & & 12 & 12 & 4305 & & $-$ & $-$ \\ 
le450\_15a & 15 & 15 & 15 & & 18 & 18 & & 17 & 17 & $-$ & & $-$ & $-$ \\ 
le450\_15b & 15 & 15 & 15 & & 17 & 17 & & 16 & 16 & $-$ & & $-$ & $-$ \\ 
le450\_25a & 25 & 25 & 25 & & 26 & 25 & & 25 & 25 & 54 & & 0 & 0 \\ 
inithx.i.1 & 54 & 54 & 54 & & 70 & 54 & & 55 & 55 & 63 & & $-$ & $-$ \\ 
inithx.i.2 & ? & 30 & 30 & & 158 & 158 & & 158 & 158 & $-$ & & $-$ & $-$ \\ 
myciel4 & 5 & 3 & 5 & 5 & 5 & 5 & 5 & 5 & 5 & 0 & 5 & 0 & 0 \\ 
myciel5 & 6 & 3 & 6 & & 9 & 6 & & 6 & 6 & 149 & & 0 & 0 \\ 
myciel6 & ? & 3 & 6 & & 11 & 7 & & 7 & 8 & $-$ & & $-$ & $-$ \\ 
flat300\_20\_0 & ? & 11 & 11 & & 81 & 81 & & 81 & 81 & $-$ & & $-$ & $-$ \\ 
latin\_sq\_10 & ? & 90 &  & & 460 & 460 & & 460 & 460 & $-$ & & $-$ & $-$ \\ 
ash331GPIA & 4 & 3 & 4 & & 8 & 8 & & 8 & 4 & $-$ & & $-$ & 1 \\ 
will199GPIA & 7 & 6 & 7 & & 9 & 9 & & 9 & 7 & $-$ & & $-$ & 2 \\ 
kneser7\_2 & 6 & 3 & 6 & 6 & 8 & 6 & 6 & 6 & 6 & 0 & 6 & 0 & 0 \\ 
kneser7\_3 & 3 & 2 & 3 & 3 & 3 & 3 & 3 & 3 & 3 & 0 & 2 & 0 & 0 \\ 
kneser9\_4 & 3 & 2 & 3 & 3 & 4 & 3 & 3 & 3 & 3 & 0 & 809 & 0 & 0 \\ 
kneser11\_5 & 3 & 2 & 3 & & 4 & 3 & & 3 & 4 & 84 & & 2128 & $-$
\end{tabular}
\end{center}
\caption{Performance of the algorithms on COLORLIB instances and Kneser graphs}
\label{INSTANCIASFINALES}
\end{table}

Table \ref{INSTANCIASFINALES} reports the final results. Columns 1-2 display the name of the instance and its equitable chromatic number.
Columns 3-5 and 6-10 show the bounds given by the initial heuristics and the bounds obtained by each algorithm after its execution.
Finally, columns 11-14 show the time needed to solve the instance, or ``$-$'' if the algorithm is not able to solve it within the limit of
two hours. Columns called ``EqDS'' and ``EqDS$^*$'' correspond to \textsc{EqDSatur} with DSATUR-CSS and ORDER1-CSS respectively.

Results for B\&C-$LF_2$ are taken from \cite{BYCBRA}. We leave blank when the instance is not mentioned in that paper. We also recall that
these results had been obtained with a slightly different platform: an 1.8 GHz AMD-Atlon machine with Linux and XPRESS 2005-a as the linear programming solver.

From the 58 evaluated instances, CPLEX has solved 38, \textsc{EqDSatur} with DSATUR-CSS has solved 29 and with
ORDER1-CSS has solved 31. However, some of the instances not solved by both versions of \textsc{EqDSatur}
(more precisely, \texttt{3-FullIns\_3}, \texttt{4-FullIns\_3} and \texttt{5-FullIns\_3}) are indeed hard to solve by enumerative schemes,
as reported in \cite{PASS}, so in our opinion \textsc{EqDSatur} presents the expected behaviour.
On the other hand, both versions of \textsc{EqDSatur} outperform CPLEX and B\&C-$LF_2$ in \texttt{queen8\_8}, and CPLEX in \texttt{myciel5}
and \texttt{queen9\_9}. In particular, the version with ORDER1-CSS outperforms B\&C-$LF_2$ in \texttt{miles750} and \texttt{miles1000}.

Let us note that DSATUR-CSS delivers a faster algorithm than ORDER1-CSS for the set of instances solved by both.
Also, it is able to solve \texttt{queen8\_12} and \texttt{kneser11\_5} in more than half an hour. Nevertheless, by using
ORDER1-CSS, \texttt{miles750}, \texttt{miles1000}, \texttt{ash331GPIA} and \texttt{will199GPIA} can be solved without difficulty.


\section{Conclusions} \label{SCONCLU}

In this paper, we present and analyze an exact DSatur-based algorithm for ECP. We propose a pruning rule based on arithmetical
properties related to equitable partitions, which has shown to be very effective. We also discuss several color and vertex
selection strategies and how to obtain lower and upper bounds of the equitable chromatic number for initializing the algorithm.
Finally, several experiments were carried out to conclude that
our approach can tackle the resolution of random graphs better than other algorithms found in the literature so far.




\end{document}